\documentclass[11pt]{article}

\usepackage{fullpage}
\usepackage{amsmath,amssymb,authblk}
\usepackage{graphicx}
\usepackage[ruled,vlined]{algorithm2e}
\usepackage{algorithmic}
\usepackage{tabularx}

\newtheorem{theorem}{Theorem}

\newtheorem{lemma}[theorem]{Lemma}

\newtheorem{observation}[theorem]{Observation}

\newtheorem{proof}[theorem]{Proof}

\newcommand{\junk}[1]{}

\newcommand{\eat}[1] {}

\newcommand{\qed} {\hfill$\Box$}

\newcommand{\opt}{{\mathrm{OPT}}}

\newcommand{\J}{{\cal J}}

\newcommand{\I}{{\cal I}}

\renewcommand{\O}{{\cal O}}

\newcommand{\ResAll} {{\sc ResAll}}
\newcommand{\ZeroOneResAll} {{\sc (0-1)-ResAll}}

\newcommand{\PResAll}{{\sc PartialResAll}}
\newcommand{\PCResAll}{{\sc PrizeCollectingResAll}}

\newcommand{\lspc}{{\sc LSPC}}
\newcommand{\smfc}{{\sc SMFC}}

\newcommand{\calI} {{\cal R}}

\newcommand{\cT} {{\cal T}}

\newcommand{\cM}{{\cal M}}

\newcommand{\cC}{{\cal C}}
\newcommand{\cJ}{{\cal J}}

\newcommand{\lptr}{{\em l-ptr}}
\newcommand{\rptr}{{\em r-ptr}}
\newcommand{\ljob}{{\em l-job}}
\newcommand{\rjob}{{\em r-job}}
\newcommand{\comment}[1]{}

\newcommand{\sh} {{\rm sh}}
\newcommand{\calL} {{\cal L}}
\newcommand{\calT} {{\cal T}}
\newcommand{\calM} {{\cal M}}
\newcommand{\cA} {{\cal A}}
\newcommand{\calA} {\cA}
\newcommand{\cB} {{\cal B}}
\newcommand{\DP} {{\rm DP}}
\newcommand{\ceil}[1] {\lceil #1 \rceil}
\newcommand{\wh}[1] {\widehat{#1}}

\title{Scheduling Resources for Executing a Partial Set of Jobs}

\author[1]{Venkatesan T. Chakaravarthy}
\author[2]{Arindam Pal}
\author[1]{Sambuddha Roy}
\author[1]{Yogish Sabharwal}
\affil[1]{IBM Research Lab, New Delhi, India\\
  \texttt{\{vechakra,sambuddha,ysabharwal\}@in.ibm.com}
}
\affil[2]{Indian Institute of Technology, New Delhi.\\
\texttt{arindamp@cse.iitd.ernet.in}
}

\begin{document}

\maketitle

\begin{abstract}
In this paper, we consider the problem of choosing a minimum cost set of resources for executing
 a specified set of jobs. 
 Each input job is an interval, determined by its start-time and end-time. Each resource is also an interval determined by its start-time and end-time; moreover, every resource has a capacity and a cost associated with it. 
We consider two versions of this problem. 

In the partial covering version, 
we are also given as input a number $k$, specifying the number of jobs that must be performed. The goal is to choose $k$ jobs and find a minimum cost set of resources to perform the chosen $k$ jobs (at any point of time the capacity of the chosen set of resources should be sufficient to execute the jobs active at that time). We present an $O(\log n)$-factor approximation algorithm for this problem. 

We also consider the prize collecting version, wherein every job also has a penalty 
associated with it. The feasible solution consists of a subset of the jobs, and a set of resources, 
to perform the chosen subset of jobs. The goal is to find a feasible solution that minimizes the 
sum of the costs of the selected resources and the penalties of the jobs that are not selected. 
We present a constant factor approximation algorithm for this problem.
\end{abstract}

\section{Introduction}
\label{sec:intro}
We consider the problem of allocating resources to schedule jobs. 
%We assume that time is uniformly divided into discrete timeslots.
Each job is specified by its start-time, end-time and its demand requirement.
%we assume that all the jobs have the same demand requirement. 
Each resource is specified by its start-time, end-time, the capacity it offers and its associated cost. 
A feasible solution is a set of resources satisfying the constraint that at any timeslot, 
the sum of the capacities offered by the resources is at least the demand required by
the jobs active at that timeslot, i.e., the selected resources must cover the jobs.
The cost of a feasible solution is the sum of costs of the resources picked in the solution.
The goal is to pick a feasible solution having minimum cost.
We call this the Resource Allocation problem ({\ResAll}).

The above problem is motivated by applications in cloud and grid computing.
Consider jobs that require a common resource such as network bandwidth or storage.
The resource may be available under different plans; for instance, it is common for network bandwidth to be priced
based on the time of the day to account for the network usage patterns during the day. 
The plans may offer different capacities of the resource at different costs.
Moreover, It may be possible to lease multiple units of the resource under some plan by paying a cost
proportional to the number of units.

Bar-Noy et al. \cite{Bar-Noy} presented a $4$-approximation algorithm for the {\ResAll} problem
(See Section 4 therein).
We consider two variants of the problem.
The first variant is the partial covering version. In this problem, the input also specifies a number $k$ and
a feasible solution is only required to cover $k$ of the jobs.
The second variant is the prize collecting version wherein each job has a penalty associated with it;
for every job that is not covered by the solution, the solution incurs an additional cost, 
equivalent to the penalty corresponding to the job.
These variants are motivated by the concept of service level agreements (SLA's), 
which stipulate that a large fraction of the client's jobs are to be completed. 
We study these variants for the case where the demands of all the jobs are uniform (say $1$ unit)
and a solution is allowed to pick multiple copies of a resource by paying proportional cost.
We now define our problems formally.

\subsection{Problem Definition}
We consider the timeline $\cT$ to be uniformly divided into discrete intervals ranging from $1$ to $T$.
We refer to each integer $1\leq t\leq T$ as a {\it timeslot}.
The input consists of a set of  {\em jobs} ${\cal J}$, and a set of {\em resources} $\calI$.

Each job $j \in {\cal J}$ is specified by an interval $I(j) = [s(j), e(j)]$, where $s(j)$ and $e(j)$ are the {\em start-time} and {\em end-time}
of the job $j$. We further assume that $s(j)$ and $e(j)$ are integers in the range $[1, T]$ for every job $j$.
While the various jobs may have different intervals associated with them, we consider all the jobs to have
{\em uniform} demand requirement, say $1$ unit. 

Further, each resource $i \in \calI$ is specified
by an interval $I(i)=[s(i),e(i)]$, where $s(i)$ and $e(i)$ are the {\em start-time} and the {\em end-time}
of the resource $i$; we assume that $s(i)$ and $e(i)$ are integers in the range $[1,T]$. 
The resource $i$ is also associated with a {\em capacity} $w(i)$ and a cost $c(i)$; we
assume that $w(i)$ is an integer. 
We interchangeably refer to the resources as {\em resource intervals}.
%For each resource $i\in \calI$, we view the interval $I(i)=[s(i),e(i)]$ as a set of timeslots in the range $[s(i),e(i)]$.
A typical scenario of such a collection of jobs and resources is shown in Figure~\ref{fig:cc}.

\begin{figure*}[t]
\begin{center}
\fbox{
\includegraphics[scale=0.25]{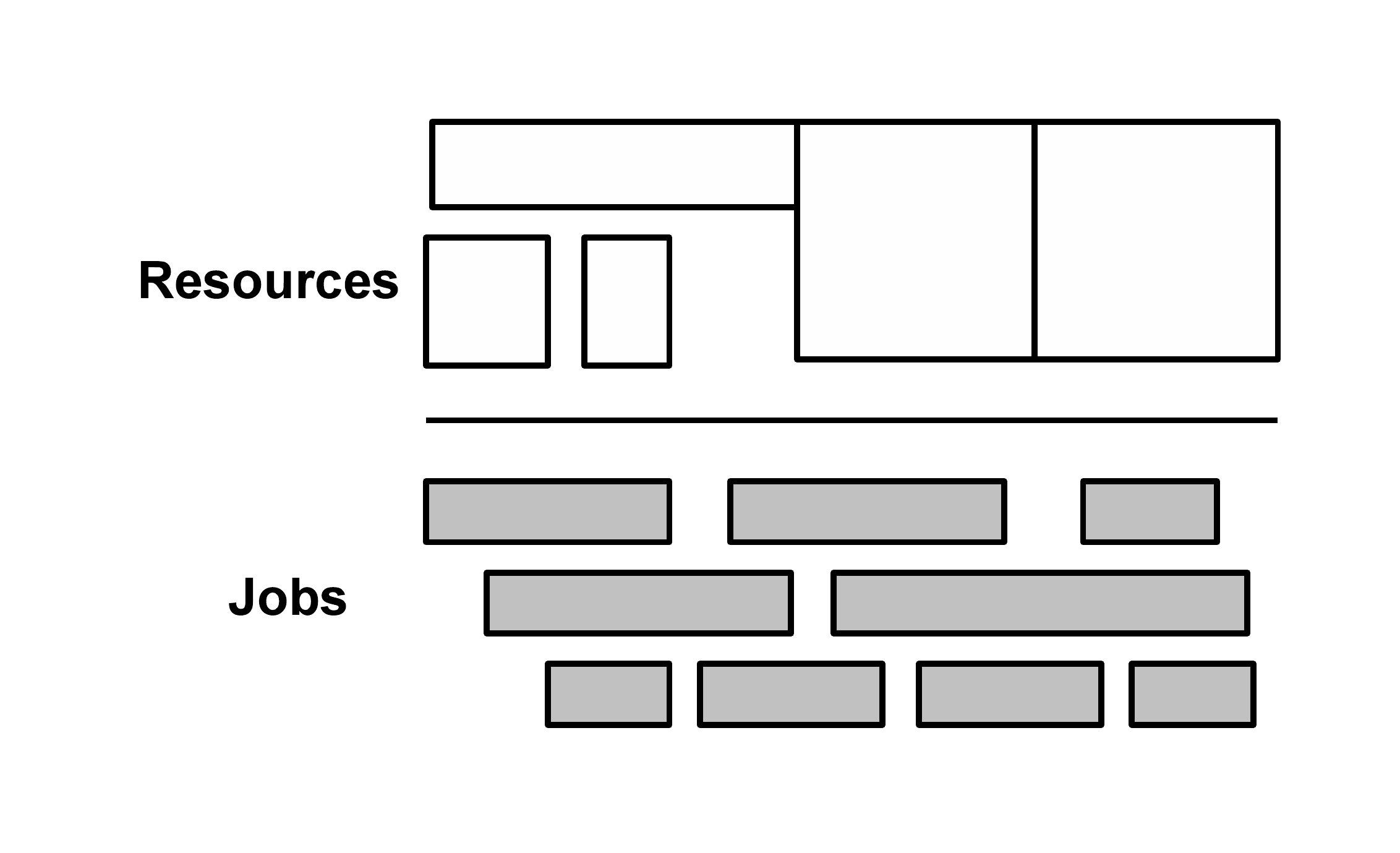}
}
\end{center}
\caption{Illustration of the input}
\label{fig:cc}
\end{figure*}

We say that a job $j$ (resource $i$) is {\it active} at a timeslot $t$, 
if $t \in I(j)$ ($I(i)$); we  denote this as $j \sim t$ ($i \sim t$).
In this case, we also say that $j$ (or $i$) {\em spans} $t$.

%Given a timeslot $t$ and a collection $R$ of resources, let $A(R,t)$ denote the set consisting of 
%the resources in $R$ that are active at timeslot $t$. 
%Analogously, 
%given a timeslot $t$ and a subset $J$ of jobs, let $A(J,t)$ denote the set of {\em jobs} from the subset $J$ active 
%at timeslot $t$.  While this is slight abuse of the notation $A(\cdot, t)$, there will be no cause for confusion, because it will be clear from the context whether the set involved consists of resources or jobs.

We define a {\em profile} $P: \cT \rightarrow \mathbb{N}$ to be a mapping that assigns an integer value 
to every timeslot. For two profiles, $P_1$ and $P_2$, $P_1$ is said to {\em cover} $P_2$, 
if $P_1(t) \geq P_2(t)$ for all $t \in \cT$.
Given a set $J$ of jobs, the profile $P_J(\cdot)$ of $J$ is defined to be the mapping determined by 
the cumulative demand of the jobs in $J$, i.e. $P_J(t) = |\{ j \in J~:~j \sim t \}|$.
Similarly, given a multiset $R$ of resources, its profile is: $P_R(t) = \sum_{i \in R~:~i \sim t} w(i)$
(taking copies of a resource into account). 
We say that $R$ {\em covers} $J$ if $P_R$ covers $P_J$. 
The cost of a multiset of resources $R$ is defined to be the sum of the costs of all the resources 
(taking copies into account).

We now describe the two versions of the problem.
\begin{itemize}
\item  {\PResAll}: In this problem, the input also specifies a number $k$ (called the {\em partiality parameter})
       that indicates the 
	number of jobs to be covered. A feasible solution is a pair $(R,J)$ where $R$ is a multiset of resources
	and $J$ is a set of jobs such that $R$ covers $J$ and $|J| \ge k$.
	The problem is to find a feasible solution of minimum cost.
\item  {\PCResAll}: In this problem, every job $j$ also has a penalty $p_j$ associated with it.
	A feasible solution is a pair $(R,J)$ where $R$ is a multiset of resources
	and $J$ is a set of jobs such that $R$ covers $J$.
	The cost of the solution is the sum of the 
	costs of the resources in $R$ and the penalties of the jobs not in $J$.
	The problem is to find a feasible solution of minimum cost.
\end{itemize}
Note that in both the versions, multiple copies of the same resource can be picked
by paying the corresponding cost as many times.

%In this paper, we reserve the uppercase letters $R, S$ for {\em sets} of resources, 
%while the uppercase letter $J$ is
%reserved for {\em sets} of jobs. Likewise, the lowercase letters $i, r$ will be reserved for individual 
%resources, while the index $j$ will be reserved for individual jobs.

\subsection{Related Work and Our Results}
Our work belongs to the space of {\em partial} covering problems, which are a 
natural variant of the corresponding full cover problems. There is a significant body of work 
that consider such problems in the literature, for instance, see \cite{Garg05,Bar01,JV01,KPS11,GKS04}.

In the setting where resources and jobs are embodied as intervals, the objective of finding a minimum cost collection of
resources that fulfill the jobs is typically called the {\em full cover} problem. Full cover problems in the interval context have 
been dealt with earlier, in various earlier works \cite{Bar-Noy,bhatia07,cgk10}. Partial cover 
problems in the interval context have been considered earlier in \cite{esa2011}.
\begin{quote}
{\bf Our Main Result}.
We present an $O(\log (n+m))$ approximation for the {\PResAll} problem, where
$n$ is the number of jobs and $m$ is the number of resources respectively.
\end{quote}
The work in existing literature that is closest in spirit to our result is that of
Bar-Noy et al.\cite{Bar-Noy}, and Chakaravarthy et al.\cite{esa2011}.
In \cite{Bar-Noy}, the authors consider the full cover version, and present 
a $4$-approximation algorithm. In this case, all the jobs are to be covered, 
and therefore the demand profile to be covered is fixed. The goal is to find the 
minimum cost set of resources, for covering this profile. 
In our setting, 
we need to cover only $k$ of the jobs. 
A solution needs to select $k$ jobs to be covered in such a manner 
that  the resources required to cover the resulting demand profile has minimum cost.  

In \cite{esa2011}, the authors consider a scenario, wherein the 
timeslots have demands and a solution must satisfy the demand for at least $k$ of the timeslots. 
In contrast, in our setting, a solution needs to satisfy $k$ {\em jobs}, wherein 
each job can span multiple timeslots. 
A job may not be completely spanned by any resource, and thus may require
{\em multiple} resource intervals for covering it.  

We also show a constant factor approximation algorithm for the {\PCResAll} problem, by reducing it 
to the zero-one version of the {\ResAll} problem. 
Jain and Vazirani \cite{JV01} provide a general framework for achieving approximation algorithms for partial 
covering problems, wherein the prize collecting version is considered. In this framework, under suitable conditions, 
a constant factor approximation for the prize collecting version implies a constant factor approximation
for the partial version as well. However, their result applies only when the prize collecting algorithm has a
 certain strong property, called the {\em Lagrangian Multiplier Preserving} (LMP) property. 
While we are able to achieve a constant factor approximation for the {\PCResAll} problem, 
our algorithm does not have the LMP property. Thus, the Jain-Vazirani framework does not apply 
to our scenario. 

\section{Outline of the Main Algorithm}
\label{sec:overview}
In this section, we outline the proof of our main result:

\begin{theorem}
\label{thm:xAAA}
There exists an $O(log (n+m))$-approximation algorithm for the {\PResAll} problem,
where $n$ is the number of jobs and $m$ is the number of resources.
\end{theorem}

The proof of the above theorem goes via the claim that the input set of jobs can be 
partitioned into a logarithmic number of {\em mountain ranges}. 
A collection of jobs $M$ is called a {\em mountain} if there exists a timeslot $t$, such that
all the jobs in this collection span the timeslot $t$; the specified timeslot where the jobs 
intersect will be called the {\em peak} timeslot of the mountain (see Figure~\ref{fig:aa};
jobs are shown on the top and the profile is shown below).
The justification for this linguistic convention is that if we look at the profile of such a 
collection of jobs, the profile forms a bitonic sequence, increasing in height until the peak, and 
then decreasing. 
The {\em span} of a mountain is the interval of timeslots where any job in the mountain is active.  
%Two mountains are said to be ``disjoint'' if their spans are disjoint.
%A collection of jobs $J$ is called a {\em mountain range}, if the jobs can be partitioned into 
%a collection of disjoint mountains (see Figure \ref{fig:bb}).
A collection of jobs $\calM$ is called a {\em mountain range}, if the jobs can be partitioned into 
a sequence $M_1, M_2, \ldots, M_r$ such that each $M_i$ is a mountain and the spans of any two mountains 
are non-overlapping (see Figure \ref{fig:bb}).
The decomposition lemma below shows that the input set of jobs can be partitioned into a logarithmic number
of mountain ranges.  For a job $j$ with 
start- and end-times $s(j)$ and $e(j)$, let its {\em length} be $\ell_j = (e(j) - s(j) +1)$). 
Let  $\ell_{\min}$ be the shortest job length, and $\ell_{\max}$ the longest job length. 
The proof of the lemma is inspired by the algorithm for the Unsplittable Flow Problem on 
a line, due to Bansal et al.~\cite{BansalFKS09}, and it is given in Appendix \ref{sec:DDD}.

\begin{figure}[t!]
\begin{minipage}{0.4\linewidth}
\centering
\includegraphics[width=1in]{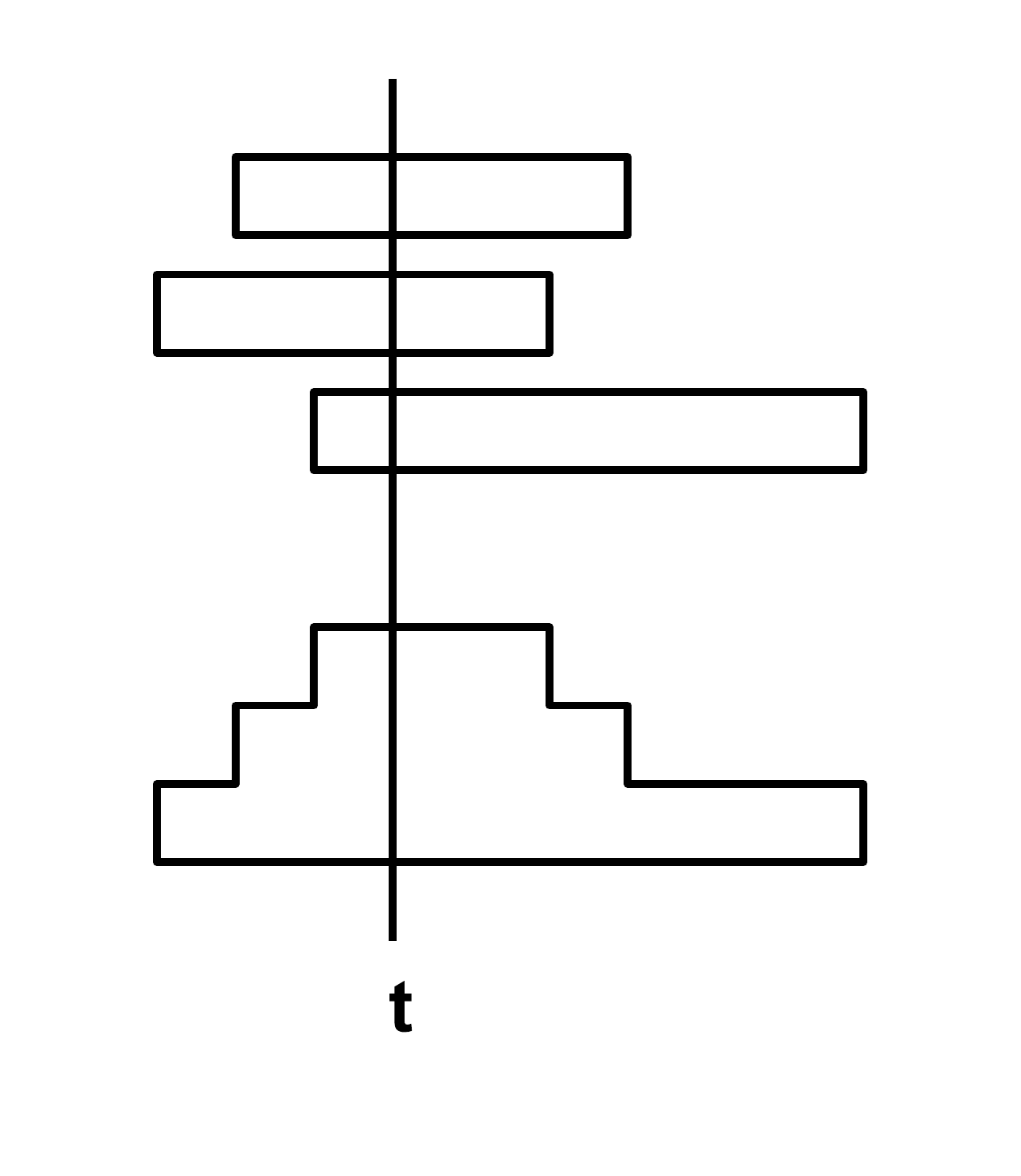}
\caption{
A Mountain $M$
}
\label{fig:aa}
\end{minipage}
%\hspace*{0.5cm}
\centering
\begin{minipage}{0.55\linewidth}
\includegraphics[width=2.5in]{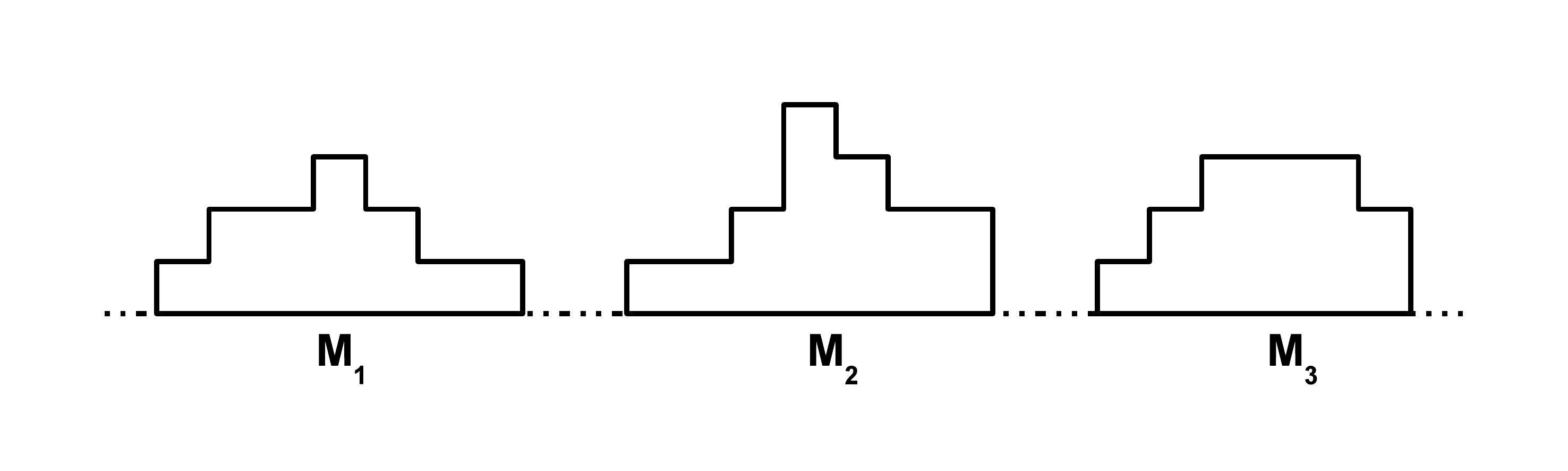}
\caption{
A Mountain Range ${\cal M}=\{M_1, M_2, M_3\}$
}
\label{fig:bb}
\end{minipage}
\end{figure}

\begin{lemma}
\label{lem:XXX}
The input set of jobs can be partitioned into groups, $\calM_1, \calM_2, \ldots, \calM_L$, such that 
each $\calM_i$ is a mountain range and $L \leq 4\cdot\lceil{\log \frac{\ell_{\max}}{\ell_{\min}}}\rceil$.
\end{lemma}

Theorem~\ref{thm:xCCC} (see below) provides a $c$-approximation algorithm (where $c$ is a constant) for the special
case where the input set of jobs form a single mountain range. 
We now prove Theorem \ref{thm:xAAA}, assuming Lemma \ref{lem:XXX} and Theorem \ref{thm:xCCC}. 

\par\noindent
{\bf Proof of Theorem~\ref{thm:xAAA}.}
Let $\cJ$ be the input set of jobs, $\calI$ be the input set of resources and $k$ be the partiality parameter.
Invoke Lemma \ref{lem:XXX} on the input set of jobs $\cJ$ and obtain a partitioning of $\cJ$
into mountain ranges $\calM_1, \calM_2, \ldots, \calM_{L}$, where $L=4\cdot \ceil{\log(\ell_{\max}/\ell_{\min})}$.
Theorem \ref{thm:xCCC} provides a $c$-approximation algorithm $\cA$ for the {\PResAll} problem wherein
the input set of jobs form a single mountain range, where $c$ is some constant.
We shall present a $(cL)$-approximation algorithm for the {\PResAll} problem.

For $1\leq q \leq L$ and $1\leq \kappa \leq k$, let $\cA(q, \kappa)$ denote the 
cost of the (approximately optimal) solution returned by the algorithm in Theorem \ref{thm:xCCC}
with $\calM_q$ as the input set of jobs, $\calI$ as the input set of resources and $\kappa$
as the partiality parameter.
Similarly, let $\opt(q, \kappa)$ denote the cost of the optimal solution for covering $\kappa$ of the jobs in the 
mountain range $\calM_q$. Theorem~\ref{thm:xCCC} implies that $\cA(q,\kappa) \leq c\cdot \opt(q,\kappa)$.

The algorithm employs dynamic programming. 
We maintain a $2$-dimensional DP table $\DP[\cdot,\cdot]$. 
For each $1\leq q\leq L$ and $1\leq \kappa \leq k$,
the entry $\DP[q,\kappa]$ would 
store the cost of a (near-optimal) feasible solution
covering $\kappa$ of the jobs from $\calM_1 \cup \calM_2 \cup \cdots \cup \calM_q$.
The entries are calculated as follows.
\begin{equation*} 
\DP[q,\kappa] = \min_{\kappa' \leq \kappa} \{\DP[q-1, \kappa - \kappa'] + \cA(q,\kappa')\}.
\end{equation*}

The above recurrence relation considers covering $\kappa'$ jobs from the mountain $\calM_q$, 
and the remaining $\kappa - \kappa'$  jobs from the mountain ranges $\calM_1, \cdots, \calM_{q-1}$.
Using this dynamic program, we compute a feasible solution to the original problem instance 
(i.e., covering $k$ jobs from all the mountain ranges $\calM_1,\calM_2, \ldots, \calM_L$);
the solution would correspond to the entry $\DP[L,k]$.
%However, note that in this solution, we use multiple copies of the resources; for instance resources might be common 
%across $\cA(q, \kappa')$ and $\DP[q-1, \kappa-\kappa']$ in the above. 
Consider the optimum solution $\opt$ to the original problem instance. 
Suppose that $\opt$ covers $k_q$ jobs from the mountain range $\calM_q$ (for $1\leq q \leq L$), 
such that $k_1 + k_2 +\cdots + k_L = k$. 
Observe that 
\begin{eqnarray*}
\DP[L,k] 
&\leq& \sum_{q=1}^L \calA(q, k_q)\\
&\leq& c\cdot{\sum_{q=1}^L \opt(q, k_q)},
\end{eqnarray*}
where the first statement follows from the construction of the dynamic programming table
and the second statement follows from the guarantee given by algorithm $\calA$.
However the maximum of $\opt(q, k_q)$ (over all $q$) is a lower bound for $\opt$
(we cannot say anything stronger than this since $\opt$ might 
use the same resources to cover jobs across multiple subsets $\calM_q$). 
This implies that $\DP[L,k] \leq c\cdot L \cdot\opt$. This proves
the $(cL)$-approximation ratio. 

\begin{theorem}
\label{thm:xCCC}
There exists a constant factor approximation algorithm for the special case of the {\PResAll} problem,
wherein the input set of jobs form a single mountain range $\calM$.
\end{theorem}

The first step in proving the above theorem is to design an algorithm for handling the special case where the input set of jobs
form a single mountain. This is accomplished by the following theorem. The proof is given in 
Section \ref{sec:mountain}.

\begin{theorem}
\label{thm:xDDD}
There exists an $8$-approximation algorithm for the special case of the {\PResAll} problem
wherein the input set of jobs for a single mountain $M$.
\end{theorem}

We now sketch the proof of Theorem \ref{thm:xCCC}. 
Let the input mountain range be ${\calM}$ consisting of mountains $M_1,M_2, \ldots, M_r$.
The basic intuition behind the algorithm is to ``collapse'' each mountain $M_q$ into a single timeslot. 
A resource interval $i$ is said to intersect a mountain $M$ if the interval $i$ and the span of 
$M$ overlap; the resource $i$ is said to {\em fully span} the mountain $M$, if 
the span of $M$ is contained in the interval $i$; the resource $i$ is said to
be contained in the mountain $M$, if the interval $i$ is contained in the span of $M$.
It may be possible that for a resource interval $i$ and a mountain $M$,
neither $i$ fully spans $M$ nor is $i$ contained in $M$.
However, at a factor three loss in the approximation ratio, we can transform an input instance into an instance
satisfying the following property. The resource intervals in the modified instance can be 
classified into two categories: (1) {\em narrow} resources $i$ having the property
that the interval $i$ is contained in the span of a specific single mountain $M$;
(2) {\em wide} resources $i$ having the property that 
if $i$ intersects any mountain $M$, then it fully spans the mountain. 

The notion of collapsing mountains into timeslots is natural when the input instance 
consists only of wide resources. 
This is because we can collapse the mountains $M_1, M_2, \ldots, M_r$ into timeslots $1, 2, \ldots, r$.  
Furthermore, for each wide resource $i$, consider the sequence of mountains 
$M_p, M_{p+1}, \ldots, M_q$ (for some $p \leq q$) that are fully spanned by the resource $i$;
then we represent $i$ by an interval that spans the timeslots $[p,q]$. 
However, the case of narrow resources is more involved
because a narrow resource does not fully span the mountain containing it. 
Based on the above intuition, we define a problem called the {\em Long Short Partial Cover} (\lspc).
The algorithm for handling a mountain range goes via a reduction to the {\lspc} problem.

{\it  Problem Definition (\lspc):} 
We are given a demand profile over a range $[1,T]$, 
which specifies an integral demand $d_t$ at each timeslot $t\in [1,T]$.
The input resources are of two types, {\em short} and {\em long}. 
A short resource spans only one timeslot, whereas a long resource can span one or more timeslots.
Each resource $i$ has a cost $c(i)$ and a capacity $w(i)$. 
The input also specifies a {\em partiality parameter} $k$. 
A feasible solution $S$ consists of a multiset of resources $S$ and a coverage profile.
A {\em coverage profile} is a function that assigns
an integer $k_t$ for each timeslot $t$ satisfying $k_t \leq d_t$.
The solution should have the following properties: 
(i) $\sum_t k_t \geq k$;
(ii) at any timeslot $t$, the sum of capacities of the resource intervals from $S$ active at $t$ is at least $k_t$;
(iii) for any timeslot $t$, at most one of the short resources spanning the timeslot $t$ 
is picked (however, multiple copies of a long resource may be included). 
The objective is to find a feasible solution having minimum cost. 
See Figure~\ref{fig:dd} for an example (in the figure, short resources are shaded).

\begin{figure*}[t]
\begin{center}
\fbox{
\includegraphics[width=1.5in]{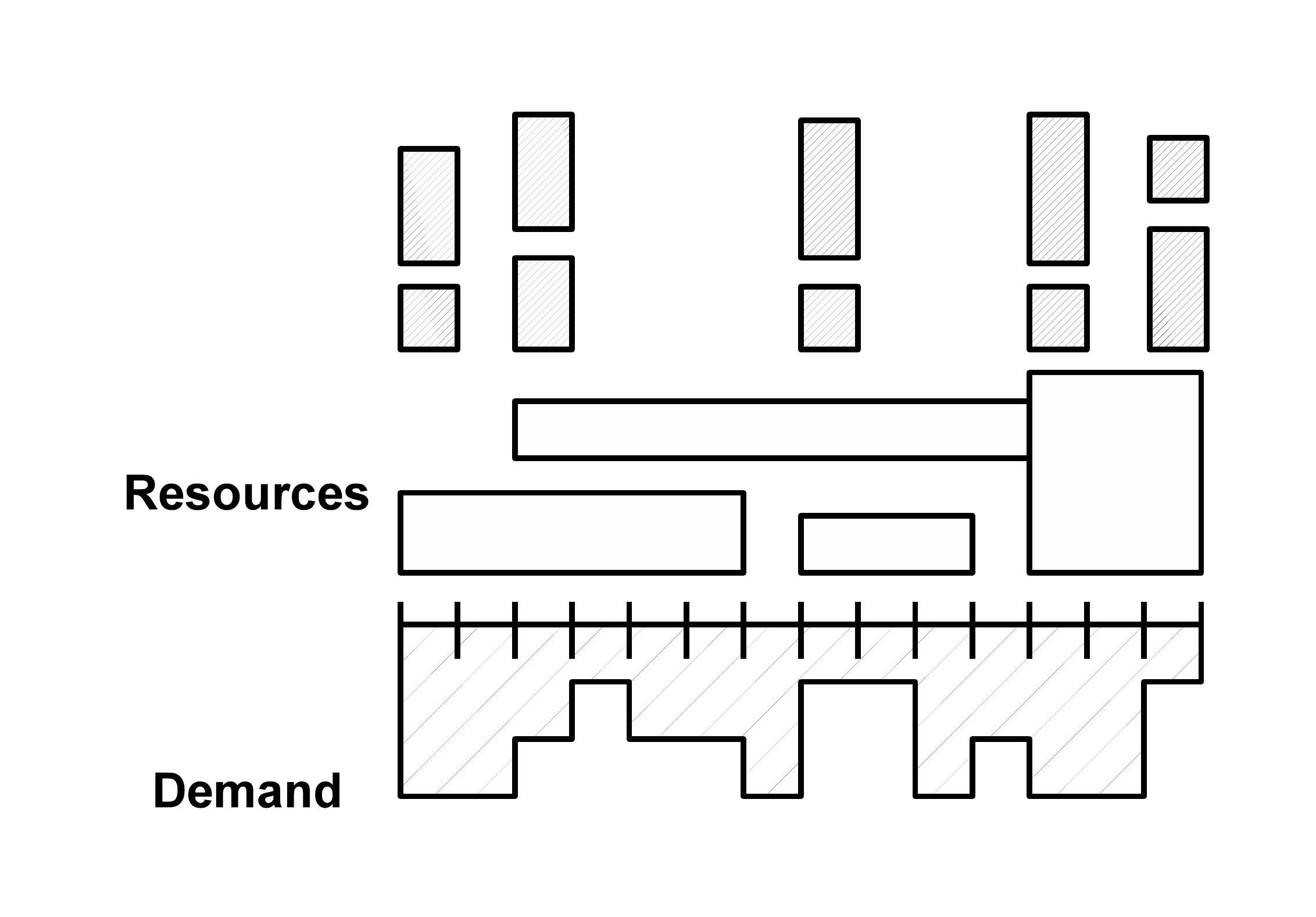}
}
\end{center}
\caption{The {\lspc} problem}
\label{fig:dd}
\end{figure*}

The advantage with the {\lspc} problem is that the demands are restricted to single timeslots; in contrast,
in the {\PResAll} problem, the demands or jobs can span multiple timeslots.
Theorem \ref{thm:xEEE} (see below) shows that the {\lspc} problem can be approximated within a factor of $16$.
The reduction from the {\PResAll} problem restricted to a single mountain range 
(as in Theorem \ref{thm:xCCC}) to the {\lspc} problem goes by representing each mountain in the input mountain
range $\calM$ by a single timeslot in the {\lspc} instance; 
the wide resources will correspond to long resources in the {\lspc} instance.
%Unfortunately, we cannot directly represent a narrow resource $i$ in the input instance.
The reduction handles the narrow resources using the short resources; the
constraint (iii) in the {\lspc} problem definition is crucially employed in this process.
The reduction from the case of single mountain range to the {\lspc} problem 
shown in Appendix~\ref{app:red} and Theorem~\ref{thm:xCCC} is proved there.

\begin{theorem}
\label{thm:xEEE}
There exists a $16$-approximation algorithm for the {\lspc} problem.
\end{theorem}

The algorithm claimed in the above theorem is inspired by the work of \cite{esa2011}. In that paper, the authors
study a variant of the problem; in that variant,
there are only long resources and a solution $S$ must satisfy a set of $k$
timeslots $t_1,t_2,\ldots, t_k\in [1,T]$, where a timeslot $t$
is satisfied, if the sum of capacities of the resources in $S$ active at $t$
is at least the demand $d_t$; a solution is allowed to pick multiple copies of any resource (both long and short). 
The {\lspc} problem differs in two ways: first, a solution can satisfy the demand at a timeslot partially
and secondly, only one copy of a short resource can be picked.
These two differences give rise to complications and as a result, our algorithm is more involved.
The algorithm is provided in Section \ref{sec:lspc}.\\

\noindent
{\bf Organization of the paper:}
Lemma~\ref{lem:XXX} is proved in Appendix~\ref{sec:DDD}.
Theorem~\ref{thm:xDDD} and Theorem~\ref{thm:xEEE} are
proved in Section~\ref{sec:mountain} and Section~\ref{sec:lspc} respectively.
Assuming Theorem~\ref{thm:xEEE}, we prove Theorem~\ref{thm:xCCC} 
in Appendix~\ref{app:red}. 

\section{A Single Mountain: Proof of Theorem~\ref{thm:xDDD}}
\label{sec:mountain}
In this section, we give an $8$-factor approximation algorithm for the case of the {\PResAll} problem,
where the input jobs form a single mountain. 

The basic intuition is as follows. Given the structure of the jobs, we will show that there is a
{\em near-optimal} feasible solution that exhibits a nice property: 
the jobs discarded from the solution are extremal either in their
start-times or their end-times. 

\begin{lemma}
\label{BBB}
Consider the {\PResAll} problem for a single mountain.
Let $\J = \{ j_1, j_2, \ldots, j_n \}$ be the input set of jobs.
Let $S=(R_S, J_S)$ be a feasible solution such that $R_S$ covers the set of jobs $J_S$ with $|J_S|=k$. 
Let $C_S$ denote its cost.
Let $L = < l_1, l_2, \ldots, l_n >$ denote the jobs in increasing order of their start-times.
Similarly, let $R = < r_1, r_2, \ldots, r_n >$ denote the jobs in decreasing order of their end-times.
Then, there exists a feasible solution $X=(R_X, J_X)$ having cost at most $2 \cdot C_S$ such that
\begin{equation}\label{eqn:aa}
{\J} \setminus J_X = \{ {l_i} : i \le q_1 \} \cup \{ {r_i} : i \le q_2 \}
\end{equation}
for some $q_1,q_2 \ge 0$ where $|\cJ \setminus J_X|=n-k$.
\end{lemma}
\begin{proof}
%Note that $k$ is the number of jobs selected in solution $S$.
We give a constructive proof to determine the sets $J_X$ and $R_X$.
%This is described in the algorithm of Figure~\ref{sd-algo}.

We initialize the set $J_X$=${\cal J}$. At the end of the algorithm, 
the set $J_X$ will be the desired set of jobs covered by the solution.
The idea is to remove the jobs that extend most to the right or the left from the consideration of $J_X$.
The most critical aspect of the construction is to ensure that whenever we exclude any job from consideration
of $J_X$ that is already part of $J_S$, we do so in pairs of the leftmost and rightmost extending jobs of $J_S$
that are still remaining in $J_X$. 
We terminate this process when the size of $J_X$ equals the size of $J_S$, i.e., $k$.
We also initialize the set $U=\phi$. 
At the end of the algorithm, this set will contain the set of jobs removed from ${\J}$ that belonged to $J_S$
while constructing $J_X$.

We now describe the construction of $J_X$ formally.
We maintain two pointers {\lptr} and {\rptr}; {\lptr} indexes the jobs in the sequence $L$ 
and {\rptr} indexes the jobs in the sequence $R$.
We keep incrementing the pointer {\lptr} and removing the corresponding job from $J_X$ 
(if it has not already been removed) until either the size of $J_X$ reaches $k$ 
or we encounter a job (say {\ljob}) in $J_X$ that belongs to $J_S$;  we do not yet remove the job {\ljob}.
We now switch to the pointer {\rptr} and start incrementing it and removing the corresponding job from $J_X$ 
(if it has not already been removed) until either the size of $J_X$ reaches $k$ 
or we encounter a job (say {\rjob}) in $J_X$ that belongs to $J_S$;  we do not yet remove the job {\rjob}.
If the size of $J_X$ reaches $k$, we have the required set $J_X$. 

Now suppose that $|J_X| \ne k$.
Note that both {\lptr} and {\rptr} are pointing to jobs in $J_S$.
Let {\ljob} and {\rjob} be the jobs pointed to by {\lptr} and {\rptr} respectively (note that 
these two jobs may be same).

We shall remove one or both of {\ljob} and {\rjob} from $J_X$ and put them in $U$. 
We classify these jobs into three categories: {\em single}, {\em paired} and {\em artificially paired}.

Suppose that $|J_X| \ge k+2$.
In this case, we have to delete at least 2 more jobs; so we delete both {\ljob} and {\rjob} and add them to $U$
as {\em paired} jobs.
In case {\ljob} and {\rjob} are the same job, we just delete this job and add it to $U$ as a {\em single} job.
We also increment the {\lptr} and {\rptr} pointers to the next job indices in their respective sequence.
We then repeat the same process again, searching for another pair of jobs.

Suppose that $|J_X|=k+1$.
In case {\ljob} and {\rjob} are the same job, we just delete this job and get the required set $J_X$ of size $k$;
We add this job to the set $U$ as a {\em single} job.
On the other hand, if {\ljob} and {\rjob} are different jobs,
we remove {\ljob} from $J_X$ and add it to $U$ as {\em artificially paired} with
its pair as the job {\rjob} ; note that we do not remove {\rjob} from $J_X$.

This procedure gives us the required set $J_X$.
We now construct $R_X$ by simply doubling the resources of $R_S$; meaning, that for each 
resource in $R_S$, we take twice the number of copies in $R_X$. 
Clearly $C_X = 2 \cdot C_S$.
It remains to argue that $R_X$ covers $J_X$.
For this, note that $U=J_S-J_X$ and hence $|U|=|J_X-J_S|$ (because $|J_X|=|J_S|=k$).
We create an arbitrary bijection $f : U \rightarrow J_X-J_S$.
Note that $J_X$ can be obtained from $J_S$ by deleting the jobs in $U$ and adding the jobs of $J_X - J_S$.

We now make an important observation:
\begin{observation}
\label{obs1}
For any {\em paired} jobs or {\em artificially paired} jobs $j_1$, $j_2$ added to $U$, 
all the jobs in $J_X$ are contained within the
span of this pair, i.e., for any $j$ in $J_X$, $s_j \ge \min \{ s({j_1}), s({j_2}) \}$ and $e_j \le \max \{ e({j_1}), e({j_2}) \}$.
Similarly for any {\em single} job $j_1$ added to $U$, all jobs in $J_X$ are contained in the span of $j_1$.
\end{observation}

For every {\em paired} jobs, $j_1$, $j_2$, Observation~\ref{obs1} implies that taking 2 copies of the 
resources covering $\{ j_1, j_2 \}$ suffices to cover $\{ f({j_1}), f({j_2}) \}$.
Similarly, for every {\em single} job $j$, the resources covering $\{ j \}$ suffice to cover $\{ f(j) \}$.
Lastly for every {\em artificially paired} jobs $j_1, j_2$ where $j_1 \in U$ and $j_2 \notin U$, taking 2 copies
of the resources covering $\{ j_1, j_2 \}$ suffices to cover $\{ f({j_1}), j_2 \}$.

Hence the set $R_X$ obtained by doubling the resources $R_S$ (that cover $J_S$) suffices to cover the jobs in $J_X$.
%\qed
\end{proof}

Recall that Bar-Noy et al.~\cite{Bar-Noy} presented a $4$-approximation algorithm for 
the {\ResAll} problem (full cover version). Our algorithm for handling a single mountain works
as follows. 
Given a mountain consisting of the collection of jobs $\J$ and the number $k$, 
do the following for all possible pairs of numbers $(q_1, q_2)$ such that the set 
$J_X$ defined as per Equation~\ref{eqn:aa} in Lemma~\ref{BBB} has size $k$.
For the collection of jobs $J_X$, consider the issue of selecting a minimum cost set of 
resources to cover these jobs; note that this is a full cover problem. Thus, the $4$-approximation
of \cite{Bar-Noy} can be applied here. 
Finally, we output the best solution across 
all choices of $(q_1, q_2)$.  Lemma~\ref{BBB} shows that this is an $8$-factor approximation
to the {\PResAll} problem for a single mountain. 

\section{{\lspc} Problem: Proof of Theorem \ref{thm:xEEE}}
\label{sec:lspc}
Here, we present a $16$-approximation algorithm for the {\lspc} problem.

We extend the notion of profiles and coverage to ranges contained within $[1,T]$. 
Let $[a,b]$ contained in $[1,T]$ be a timerange. 
By a profile over $[a,b]$, we mean a function $Q$ that assigns a value $Q(t)$ to each timeslot $t\in [a,b]$. 
A profile $Q$ defined over a range $[a,b]$ is said to be {\em good}, if for all timeslots $t\in [a,b]$,
$Q(t)\leq d_t$ (where $d_t$ is the input demand at $t$).
In the remainder of the discussion, we shall only consider good profiles and so, we shall simply write
``profile'' to mean a ``good profile''.
The {\em measure} of $Q$ is defined to be the sum  $\sum_{t\in [a,b]} Q(t)$.

Let $S$ be a multiset of resources and let $Q$ be a profile over a range of timeslots $[a,b]$.
We say that $S$ is {\em good}, if it includes at most one short resource active at any timeslot $t$.
We say that $S$ covers the profile $Q$, 
if for any timeslot $t\in [a,b]$, the sum of capacities of resources
in $S$ active at $t$ is at least $Q(t)$. 
Notice that $S$ is a feasible solution to the input problem instance,
if there exists a profile $Q$ over the entire range $[1,T]$ such that 
$Q$ has measure $k$ and $S$ is a cover for $Q$.
For a timeslot $t\in [1,T]$, let $Q^{\sh}_S(t)$
denote the capacity of the unique short resource from $S$ active at $t$, if one exists; otherwise, 
$Q^{\sh}_S(t) = 0$.

Let $S$ be a good multiset of resources and let $Q$ be a profile over a range of timeslots $[a,b]$. 
For a long resource $i\in S$, let $f_S(i)$ denote the number of copies of $i$ included in $S$.
The multiset $S$ is said to be a {\em single long resource assignment cover} (SLRA cover) for $Q$,
if for any timeslot $t\in [a,b]$, there exists a long resource $i\in S$ such that
$w(i)f_S(i)\geq Q(t)-Q^{\sh}_S(t)$ (intuitively, the resource $i$ can cover the residual demand by itself,
even though other long resources in $S$ may be active at $t$).

We say that a good multiset of resources $S$ is an {\em SLRA solution} to the input {\lspc} problem instance,
if there exists a profile $Q$ over the range $[1,T]$ having measure $k$ such that $S$ is an SLRA cover for $Q$.
The lemma below shows that near-optimal SLRA solutions exist.

\begin{lemma}
\label{lem:SLRA}
Consider the input instance of the {\lspc} problem.
There exists an SLRA solution having cost at most 16 times the cost of the optimal solution.
\end{lemma}

The lemma follows from a similar result proved in \cite{esa2011} and
the proof is deferred to Appendix~\ref{sec:lspcdetails}.
Surprisingly, we can find the {\em optimum} SLRA solution $S^*$ in polynomial time, 
as shown in Theorem~\ref{thm:xGGG} below.
Lemma \ref{lem:SLRA} and Theorem \ref{thm:xGGG} imply that 
$S^*$ is a $16$-factor approximation to the optimum solution.
This completes the proof of Theorem~\ref{thm:xEEE}.

\begin{theorem}
\label{thm:xGGG}
The optimum SLRA solution $S^*$ can be found in time polynomial in the number of resources,
number of timeslots and $H$, where $H=\max_{t\in[1,T]} d_t$.
\end{theorem}

The rest of the section is devoted to proving Theorem~\ref{thm:xGGG}.
The algorithm goes via dynamic programming.
The following notation is useful in our discussion.
\begin{itemize}
\item
Let $S$ be a good set consisting of only short resources, and let $[a,b]$ be a range.
For a profile $Q$ defined over $[a,b]$, 
 $S$ is said to be an {\em $h$-free cover} for $Q$, if for any $t\in [a,b]$,
$Q^{\sh}_S(t)\geq Q(t)-h$.
The set $S$ is said to be an {\em $h$-free $q$-cover} for $[a,b]$,
if there exists a profile $Q$ over $[a,b]$  such that $Q$ has measure $q$ and $S$
is a $h$-free cover for $Q$.
\item 
Let $S$ be a good multiset of resources and let $[a,b]$ be a range.
For a profile $Q$ defined over $[a,b]$, 
the multiset $S$ is said to be an {\em $h$-free SLRA cover} for $Q$, 
if for any timeslot $t\in [a,b]$ satisfying $Q(t)-Q^{\sh}_S(t) > h$, 
there exists a long resource $i\in S$ such that $w(i)f_S(i)\geq Q(t)-Q^{\sh}_S(t)$.
For an integer $q$, we say $S$ is an {\em $h$-free SLRA $q$-cover} for the range $[a,b]$,
if there exists a profile $Q$ over $[a,b]$ such that $Q$ has measure $q$ and $S$
is a $h$-free SLRA cover for $Q$.
\end{itemize}
Intuitively, $h$ denotes the demand covered by long resources already selected (and their cost accounted for)
in the previous stages of the algorithm; thus, 
 timeslots whose residual demand is at most $h$ can be ignored.
The notion of ``$h$-freeness'' captures this concept.

We shall first argue that any $h$-free SLRA cover $S$ for a profile $Q$ over a timerange $[a,b]$ exhibits 
certain interesting decomposition property.
Intuitively, in most cases, the timeline can be partitioned into two parts (left and right),
and $S$ can be partitioned into two parts $S_1$ and $S_2$ such that 
$S_1$ can cover the left timerange and $S_2$ can cover the right timerange
(even though resources in $S_1$ may be active in the right timerange and those in $S_2$
may be active in the left timerange).
In the cases where the above decomposition is not possible,
there exists a long resource spanning almost the entire range.
The lemma is similar to a result proved in \cite{esa2011} (see Lemma 4 therein).
The proof is deferred to Appendix~\ref{sec:lspcdetails}.

\begin{lemma}
\label{lem:decomp}
Let $[a,b]$ be any timerange, $Q$ be a profile over $[a,b]$ and let $h$ be an integer.
Let $S$ be a good set of resources providing an $h$-free SLRA-cover for $Q$.
Then, one of the following three cases holds:
\begin{itemize}
\item
The set of short resources in $S$ form a $h$-free cover for Q.
\item 
{\it Time-cut: } There exists a timeslot $a\leq t^*\leq b-1$ and a partitioning of $S$ into $S_1$ and $S_2$
such that $S_1$ is an $h$-free SLRA-cover for $Q_1$ and $S_2$ is an $h$-free SLRA-cover for $Q_2$,
where $Q_1$ and $Q_2$ profiles obtained by restricting $Q$ to $[a,t^*]$ and $[t^*+1,b]$, respectively.
\item
{\it Interval-cut:}
There exists a long resource $i^*\in S$
such that the set of short resources in $S$ forms a $h$-free cover for both $Q_1$ and $Q_2$, where
$Q_1$ and $Q_2$  are the profiles obtained by restricting $Q$ to $[a,s(i^*)-1]$ and  $[e(i^*)+1,b]$ 
respectively.
\end{itemize}
\end{lemma}

We now discuss our dynamic programming algorithm.
Let $H=\max_{t\in[1,T]} d_t$ be the maximum of  the input demands. 
The algorithm maintains a table $M$ with an entry for each triple $\langle [a,b], q, h\rangle$,
where $[a,b]\subseteq [1,T]$, $0\leq q\leq k$ and $0\leq h\leq H$.
The entry $M([a,b],q,h)$ stores the cost of the optimum $h$-free SLRA $q$-cover for the range $[a,b]$;
if no solution exists, then $M([a,b],q,h)$ will be $\infty$.
Our algorithm outputs the solution corresponding to the entry  $M([1,T],k,0)$; notice that 
this is optimum SLRA solution $S^*$.

In order to compute the table $M$, we need an auxiliary table $A$.
For a triple $[a,b]$, $q$ and $h$, let $A([a,b],q,h)$ be the optimum $h$-free $q$-cover for $[a,b]$,
(using only the short resources); if no solution exists $A([a,b],q,h)$ is said to be $\infty$.
We first describe how to compute the auxiliary table $A$. 
For a triple consisting of $t\in [1,T]$, $q\leq k$ and
$h\leq H$, define $\gamma(t,q,h)$ as follows.
If $q > d_t$, set $\gamma(t,q,h)=\infty$.
Consider the case where $q\leq d_t$.
If $q \leq h$, set $\gamma(t,q,h) = 0$.
Otherwise, let $i$ be the minimum cost short resource active at $t$ such that $w(i) \geq q-h$;
set $\gamma(t,q,h)=c(i)$; if no such short resource exists, set $\gamma(q,t,h)=\infty$.

Then, for a triple $\langle[a,b],q, h\rangle$, the entry $A([a,b],q,h)$ is governed
by the following recurrence relation. Of the demand $q$ that need to be covered,
the optimum solution may cover a demand  $q_1$ from the timeslot $t$, and a demand $q-q_1$ from the  range $[a,b-1]$.
We try all possible values for $q_1$ and choose the best:
\[
A([a,b],q,h) = 
\min_{
  \substack{
     q_1\leq \min\{q, d_b\}
   }
}
A([a,b-1],q -q_1,h) + \gamma(b,q_1,h).
\]
It is not difficult to verify the correctness of the above recurrence relation.

\begin{figure*}
\fbox{
\begin{minipage}{0.95\textwidth}
\begin{eqnarray*}
E_1 &=& A([a,b],q,h).\\
E_2 &=& \min_{\substack{t\in [a,b-1] \\ q_1\leq q}} M([a,t],q_1,h) + M([t+1,b],q-q_1,h).\\
E_3 &=&
\quad
\min_{
\substack{(i\in \calL, \alpha\leq H)~:~\alpha w(i)>h\\
           q_1,q_2,q_3~:~q_1+q_2+q_3 = q
         }
}
\begin{pmatrix}
\alpha \cdot c(i) \\
+ A([a,s(i)-1],q_1,h) \\
+M([s(i),e(i)],q_2,\alpha w(i))\\
+ A([e(i)+1,b],q_3,h) \\
\end{pmatrix}
\end{eqnarray*}
\end{minipage}
}
\caption{Recurrence relation for $M$}
\label{fig:formula}
\end{figure*}

We now describe how to compute the table $M$. 
Based on the decomposition lemma (Lemma \ref{lem:decomp}), we can develop a recurrence relation 
for a triple $[a,b]$, $q$ and $h$. 
We compute $M([a,b],q,h)$ as the minimum over three quantities $E_1$, $E_2$ and $E_3$ corresponding to the three cases
of the lemma. Intuitive description of the three quantities is given below and precise formulas are provided
in Figure \ref{fig:formula}. In the figure, $\calL$ is the set of all long resources\footnote{The input demands
$d_t$ are used in computing the table $A(\cdot,\cdot,\cdot)$}.
\begin{itemize}
\item 
{\it Case 1: } 
No long resource is used and so, we just use the 
corresponding entry of the table $A$. 
\item
{\it Case 2: } There exists a time-cut $t^*$. We consider all possible values
of $t^*$. For each possible value, we try all possible ways in which $q$ can be divided between the left and right ranges.
\item
{\it Case 3: }
There exists a long resource $i^*$ such that the timeranges to the left of and to the right of
$i^*$ can be covered solely by short resources.
We consider all the long resources $i$ and also the number of copies $\alpha$ to be picked.
%Let $\alpha^* =f_S(i^*)$ be the number of copies of $i^*$ present in $S$.
%The input value $h$ means that in the residual profile, we can ignore timeslots having residual demand at most $h$.
%Notice that if $\alpha^* w(i^*)\leq h$, then $i^*$ is not a useful resource,
%because $i^*$ can cover in an SLRA fashion only the timeslots in $[s(i^*),e(i^*)]$ with residual demands at most $h$;
%but all such timeslots are free and need not be covered.
%So, without loss of generality, assume that $\alpha^* w(i^*)>h$.
%Since $i^*$ spans the entire range $[s(i^*),e(i^*)]$, 
%the resource $i^*$ can cover the all the timeslots in the above range with 
%residual demands at most $\alpha w(i^*)$.
%Moreover, since $d_t\leq H$ (where $H$ is the maximum input demand $d_t$), we have that $\alpha^*\leq H$.
%Hence,  our algorithm will guess $i^*$ and $\alpha^*$ by trying all possible $(i,\alpha)$ such that $\alpha \cdot %w(i)>h$.
Once $\alpha$ copies of $i$ are picked, $i$ can cover all timeslots with residual demand at most
$\alpha w(i)$ in an SLRA fashion, and so the subsequent recursive calls can ignore these timeslots.
Hence, this value is passed to the recursive call.
We also consider different ways in which $q$ can be split into three parts - left, middle and right.
The left and right parts will be covered by the solely short resources and the middle part will use both 
short and long resources.
Since we pick $\alpha$ copies of $i$, a cost of $\alpha c(i)$ is added.
\end{itemize}
We set $M([a,b],q,h)=\min\{E_1, E_2, E_3\}$. For the base case: for any $[a,b]$, if $q=0$ or $h = H$,
then the entry is set to zero.

We now describe the order in which the entries of the table are filled.
Define a partial order $\prec$ as below.
For pair of triples $z=([a,b],q,h)$ and $z'=([a',b'],q',h')$,
we say that $z\prec z'$, if one of the following properties is true:
(i)$[a',b']\subseteq [a,b]$; 
(ii) $[a,b]=[a',b']$ and $q<q'$; 
(iii) $[a,b]=[a',b']$, $q=q'$ and $h>h'$.
Construct a directed acyclic graph (DAG) $G$ where the triples
are the vertices and an  edge is drawn from a triple $z$ to a triple $z'$,
if $z\prec z'$. Let $\pi$ be a topological ordering of the vertices in $G$.
We fill the entries of the table $M$ in the order of appearance in $\pi$.
Notice that the computation for any triple $z$ only refers to triples appearing
earlier than $z$ in $\pi$.

Using Lemma \ref{lem:decomp}, we can argue that the above recurrence relation correctly computes
all the entries of $M$. 
For the sake of completeness, a proof is included in Appendix~\ref{sec:recur-proof}.

\section{The {\PCResAll} problem}
\label{app:pcresall}
In this section, we consider the {\PCResAll} problem. We prove the following:

\begin{theorem}
\label{DDD}
There is a $4$-factor approximation algorithm for the {\PCResAll} problem. 
\end{theorem}
The proof proceeds by exhibiting a reduction from the {\PCResAll} problem to 
the following full cover problem.

\noindent
{\it  Problem Definition:} We are given a demand profile which specifies an integral demand $d_t$ at each timeslot $t$.
The input resources are of two types, called S-type (short for single) and M-type (short for multiple).  
A resource $i$ has 
a capacity $w(i)$, and cost $c(i)$. A valid solution consists of a multiset of resources such that
it includes at most $1$ copy of any S-type resource; however arbitrarily many copies of any M-type resource may be picked. A feasible solution $S$ is a valid solution such that at any timeslot $t$, the total 
capacity of the resources in $S$, active at the timeslot $t$, is at least the demand $d_t$ of the timeslot.
The objective is to find a feasible solution having minimum cost. 

Call this problem the Single Multiple Full Cover ({\smfc}) problem.

The full cover problem, {\ZeroOneResAll} is considered in \cite{esa2011}. 
The {\ZeroOneResAll} problem specifies demands for 
timeslots, and feasible solutions consist of a set of resources such that every timeslot is 
fulfilled by the cumulative capacity of the resources active at that timeslot. The main 
qualification is that in this problem setting, any resource may be picked up {\em at most
once}. In \cite{esa2011},  it is shown that this problem admits a $4$-factor
approximation algorithm. The {\smfc} problem easily reduces to the {\ZeroOneResAll} problem:
{\em S-type} resources may be picked up at most once, and keep copies of the {\em M-type}
resources so that it suffices to select any one of the copies. 
Thus the algorithm and the performance guarantee claimed in \cite{esa2011} also implies the following:
\begin{theorem}\label{EEE}
There is a $4$-factor approximation to the {\smfc} problem. 
\end{theorem} 

We proceed to exhibit our reduction from the {\PCResAll} problem to the {\smfc} problem. 
Given an instance $\I$ of the {\PCResAll} problem, we will construct an instance $\O$ of the 
{\smfc} problem, such that any optimal solution $\opt(\I)$ can be converted (at no extra
cost) into an optimal solution $\opt(\O)$ for the instance $\O$. Consider any job $j$
in the instance $\I$; we will create a S-type resource $r(j)$ in the instance $\O$ corresponding to
$j$. The resource $r(j)$ will have the same length, start- and end-times as the job $j$, and will 
have a cost $p_j$ (the penalty associated with job $j$). The resources in instance $\I$ 
will be labeled as M-type resources in the instance $\O$. The other parameters, such as 
demands at timeslots, are inherited by $\O$ from the instance $\I$. 

We show that 
any feasible solution $S_\I$ to the {\PCResAll} problem corresponds to a feasible solution $S_\O$ (of the same cost) for the {\smfc} problem. Let ${\cal J}'$ denote the set of jobs that are not covered by the 
solution $S_\I$ (thus, the solution pays the penalty for each of the jobs in $\cJ'$).

The multiset of resources in $S_\O$ consists of the (M-type) resources that exist in the solution $S_\I$, and 
the S-type resources $r(j)$ in $\O$ corresponding to every job $j$ in $\cJ'$. 
Any job $j$ that is actually covered by the set of resources in $S_\I$ is also covered
in the solution $S_\O$, and the resources utilized to cover the job are the same. A job $j$ that is 
not covered by the resources in $S_\I$ pays a penalty $p_j$ in the solution $S_\I$; however this 
job $j$ in $\O$ can be covered by the S-type resource $r(j)$ in the solution 
$S_\O$. Thus, the solution $S_\O$ is a feasible solution to the instance $\O$, and has cost 
equal to the cost of the solution $S_\I$. 

In the reverse direction, suppose we are given a solution $S_\O$ to the instance $\O$. 
We will convert the solution into a {\em standard} form, i.e. a solution in which if a 
S-type resource $r(j)$ (for some job $j$) is included, then this resource is used to 
cover job $j$. Suppose job $j$ is covered by some other resources in the solution 
$S_\O$, while resource $r(j)$ covers some other jobs (call this set $J'$). We can clearly {\em exchange}
the resources between job $j$ and the set of jobs $J'$ so that job $j$ is covered by
resource $r(j)$. So we may assume that the solution $S_\O$ is in standard form. 
But now, given a standard form solution $S_\O$, we can easily construct a 
feasible solution $S_\I$ for the {\PCResAll} instance $\I$: if a job $j$ in $S_\O$ 
is covered by the S-type resource $r(j)$, then in $S_\I$, this job will not be 
covered (and a penalty $p_j$ will be accrued); all jobs $j$ in $S_\O$ that are 
covered by M-type resources will be covered by the corresponding resources
in $S_\I$.

This completes the reduction, and the proof of Theorem~\ref{EEE}.

\bibliographystyle{plain}
\bibliography{papers}

\appendix
\section{Proof of Lemma \ref{lem:XXX}}
\label{sec:DDD}
We first categorize the jobs according to their  lengths into $r$ categories $C_1$, $C_2$, $\cdots, C_r$, where 
$r = \lceil \log \frac{\ell_{\max}}{\ell_{\min}}\rceil$.
The category $C_i$ consists of all the jobs with lengths in the range $[2^{i-1}{\ell_{\min}}, 2^i{\ell_{\min}})$.
Thus all the jobs in any single category have comparable lengths: 
any two jobs $j_1$ and $j_2$ in the category satisfy $\ell_{1} < 2\ell_{2}$, where 
$\ell_1$ and $\ell_2$ are the lengths of $j_1$ and $j_2$ respectively.

Consider any category $C$ and let the lengths of the jobs in $C$ lie in the range $[\alpha, 2\alpha)$.
We claim that the category $C$ can be partitioned into $4$ groups $G_0, G_1, G_2, G_3$, such that 
each $G_i$ is a mountain range. 
To see this, 
partition the set of jobs $C$ into classes $H_1, H_2, \ldots, H_q, \ldots$ where $H_q$ consists of the jobs
active at timeslot $q \cdot \alpha$.
Note that every job belongs to some class since all the jobs have
length at least $\alpha$; if a job belongs to more than one class, assign it to any one class arbitrarily.
Clearly each class $H_q$ forms a mountain. 
For $0 \le i \le 3$, let $G_i$ be the union of the classes $H_q$ satisfying $q \equiv i \mod 4$.
Since each job has length at most $2\alpha$, each $G_i$ is a mountain range.
Thus, we get a decomposition of the input jobs into $4r$ mountain ranges.
\qed

\section{Single Mountain Range: Proof of Theorem~\ref{thm:xCCC}}
\label{app:red}
In this section, we prove Theorem~\ref{thm:xCCC} via a reduction to {\lspc}. 
The reduction proceeds in two steps. 

\subsection{First Step}
Let the input instance be $\cA$, wherein the input set of jobs form a mountain range $\cM = \{M_1, M_2, \cdots, M_r\}$. We will transform the instance $\cA$ to an instance $\cB$, with some nice properties:
(1) the input set of jobs in $\cB$ also form a mountain range;
(2) every resource $i$ in the instance $\cB$ is either narrow or wide (see Section~\ref{sec:overview} for the definitions);
(3) the cost of the optimum solution for the instance $\cB$ is at most $3$ times the optimal cost for
the instance $\cA$;
(4) given a feasible solution to $\cB$, we can construct a feasible solution to $\cA$ preserving the cost. 

Consider each resource $i$ in $\cA$ and let $M_p, M_{p+1}, \cdots, M_q$ (where $1 \leq p \leq q \leq r$) be the sequence of mountains that $i$ intersects. Clearly, $i$ fully spans the mountains $M_{p+1}, \cdots, M_{q-1}$. 
We will split the resource $i$ into at most $3$ new resources $i_1, i_2, i_3$; we say that $i_1$, $i_2$ and $i_3$ 
are {\em associated with} $i$.
The resource $i_2$ will fully span the mountains $M_{p+1}, \cdots, M_{q-1}$.
The span of the resource $i_1$ is the intersection of the span of  $i$ with the mountain $M_p$. Likewise, the span of the resource $i_3$ is the intersection of the span of $i$ with the 
mountain $M_q$.
The capacities and the costs of $i_1$, $i_2$ and $i_3$ are declared to be the same as that 
of $i$. We include $i_1, i_2, i_3$ in $\cB$. 
The input set of jobs and the partiality parameter $k$, in $\cB$ are identical to that of $\cA$. 
This completes the reduction.

It is easy to see that the first two properties are satisfied by $\cB$.
Let us now consider third property .
Given any solution $S$ for the instance $\cA$, we can construct a solution $S'$ for $\cB$ as follows. 
For each copy of resource $i$ picked in $S$, include a single copy of $i_1$, $i_2$ and $i_3$ in $S'$. 
Clearly, the cost of the solution $S'$ is at most thrice that of the cost of $S$.
Regarding the fourth property, given a solution $S$ to $\cB$, we can construct a solution $S'$ to $\cA$
as follows. Consider any resource $i$ in $\cA$ and let $i_1$, $i_2$ and $i_3$ be the resources in $\cB$
associated with $i$. Let $f_1, f_2, f_3$ be the number of copies of $i_1,i_2,i_3$ picked by solution $S$.
Let $f=\max\{f_1,f_2,f_3\}$. Include $f$ copies of the resource $i$ in the solution $S'$.
It is easy to see that $S'$ is a feasible solution to $\cA$ and that the cost of $S'$ is
at most the cost of $S$.

\subsection{Second Step}
In this step we reduce the problem instance $\cB$ to an {\lspc} instance $\cC$, with the following properties:
(1) the cost of the optimum solution for the instance $\cC$ is at most $8$ times the optimal cost for
the instance $\cA$; 
(2) Given a feasible solution to $\cC$, we can construct a feasible solution to $\cB$ preserving the cost. 

\subsubsection*{Reduction}
In the instance $\cC$,  retain only the peak timeslots of the various mountains in the instance $\cB$
so that the number of timeslots in $\cC$ is the same as the number of mountains $r$ in $\cB$. 
For any  peak timeslot $t$ in the instance $\cB$, let $d_t$ be the number of jobs in $\cB$ that are
active at the timeslot $t$; we assign the demand $d_t$ to timeslot $t$ in the instance $\cC$. 
 For any wide resource $i$ in $\cB$, fully spanning mountains $M_p, M_{p+1}, \cdots, M_q$, 
create a long resource $i'$ in $\cC$ with the span $[p,q]$. The cost and capacity of $i'$ are
the same as that of $i$. 

The narrow resources in the instance $\cB$ are used to construct the short resources in the instance $\cC$ 
as follows.
Consider any specific mountain $M$ in the instance $\cB$ along with the collection of narrow resources $R$ that are 
contained in the span of $M$, and let $t$ be the peak timeslot of $M$. For any integer $\kappa$ ($1 \leq \kappa \leq d_t$), 
we will apply the algorithm implied in Theorem~\ref{thm:xDDD} for the single mountain $M$, with $\kappa$ as the 
partiality parameter, and the set of narrow resources $R$ as the only resources. 
Then, Theorem~\ref{thm:xDDD} gives us a solution of cost $C$ consisting of a multiset $R'$ of some resources in $R$, 
that covers $\kappa$ of the jobs in the mountain $M$. Corresponding to each $\kappa$, we will include a short 
resource $i_s$ in the instance $\cC$ with capacity $\kappa$, and cost $C$.  
We will call the (multi)set of narrow resources $R' \subseteq R$ in the instance $\cB$ 
as {\em associated} with the short resource $i_s$.  This completes the description of the instance $\cC$ 
of the {\lspc} problem.

\subsubsection*{Validity of the reduction}
We will now argue the validity of the reduction. 
Let us consider the first property:
the cost of the optimum solution to the instance $\cC$ has cost at most $8$ times
the cost of the optimum solution to the instance $\cB$.
The following lemma is useful for this purpose.

\begin{lemma}
\label{lem:LLL}
Let $J$ be a subset of jobs and $R$ be multiset of resources in the instance $\cB$
such that $R$ covers $J$ (note that $R$ contains only narrow or wide resources and $J$ forms a mountain range).
Let $R_1$ and $R_2$ be narrow and wide resources in $R$. Let $R_2'$ be a multiset 
constructed by picking twice the number of copies of each resource in $R_2$.
Then, $J$ can be partitioned into two sets $J_1$ and $J_2$ such that $J_1$ is solely covered by the resources in $R_1$
and $J_2$ is solely covered by the resources in $R_2'$.
\end{lemma}
\begin{proof}
For now, we assume that the mountain range comprises of a single mountain.

Let $P_R(t), P_{R_1}(t), P_{R_2}(t), P_{R_2'}(t)$ denote the profile of the resources in $R, R_1, R_2$ and $R_2'$ respectively.
Note that $P_{R_2}(t)$ is a uniform bandwidth profile having uniform height, say $h$.
Let $J_L$ be the first $h$ jobs among all the jobs in $J$ sorted in ascending ordered by their start-times. 
Similarly, 
let $J_R$ be the first $h$ jobs among all the jobs in $J$ sorted in descending order by their left end-times.
Intuitively, $J_L$ and $J_R$ correspond to the $h$ left-most and the $h$ right-most jobs in the mountain.

Let $J_L$ and $J_R$ denote the $h$ left-most and $h$ right-most jobs in the job profile $J$ respectively
(these sets may not be disjoint). Let $J_2 = J_L \cup J_R$ and $J_1 = J \setminus J_2$.
Let $P_J(t)$, $P_{J_1}(t)$ and $P_{J_2}(t)$ denote the profiles of the jobs in $J$, $J_1$ and $J_2$ respectively.

Note that the profile $P_{R_2'}(t)$ has height $2h$ throughout the span of the mountain whereas the profile 
$P_{J_2}(t)$ has height at most $2h$ at any timeslot. Thus $R_2'$ covers $J_2$.

We will now show that $R_1$ covers $J_1$.
Note that $P_{J_1}(t)=P_J(t)-P_{J_2}(t)$.
We partition the timeslots into two parts: $T_0=\{t:P_{J_1}(t)=0 \}$ and $T_{>0}=\{t:P_{J_1}(t)>0 \}$.
For the timeslots in $T_0$, there are no jobs remaining in $J_1$ for $R_1$ to cover. 
For the timeslots in $T_{>0}$, we note that $P_{J_1}(t) \le P_J(t)-h$
(because $J_2$ comprises of the left-most $h$ and right-most $h$ jobs of the mountain).
Also note that the profile $P_{R_1}(t) = P_R(t)-P_{R_2}(t) = P_R(t)-h$.
Since, $R$ covers $J$, this implies that $R_1$ is sufficient to cover $J_1$.

The proof can easily be extended to a mountain range as the mountains within a mountain range are disjoint.
%\qed
\end{proof}

Let $\opt=(R,J)$ denote the optimal solution for the problem instance $\cB$,
where $J$ is the set of jobs picked by the solution and $R$ is the set of resources covering $J$ (we have $|J|=k$).
Let $R_1$ and $R_2$ be the set of narrow and wide resources in $R$.
Apply Lemma \ref{lem:LLL} for the solution $(R,J)$ and obtain a partition of $J$ into $J_1$ and $J_2$
along with $R_1$ (covering $J_1$) and $R_2'$ (covering $J_2$).
Let $\calM = M_1, M_2, \ldots, M_r$ be the input mountain range in the instance $\cB$ with peak timeslots
$t_1, t_2, \ldots, t_r$, respectively. Consider any mountain $M_q$.
Let $k_q$ be the number of jobs picked in $J$ from the mountain $M_q$.
Let $R_{1,q}$ be the set of (narrow) resources from $R_1$ contained within the span of $M_q$.
Thus, the set of resources $R_{1,q}$ cover the set of jobs in $M_q\cap J_1$ and let $k_q' = |M_q\cap J_1|$.
Corresponding to the value $k_q'$, we would have included a short resource, say $i_q$
in the instance $\cC$; cost of $i_q$ is at most $8$ times the cost of $R_{1,q}$ 
(as guaranteed by Theorem \ref{thm:xDDD}).
The set of long resources in $R_2'$ cover at least $k_q-k_q'$ jobs within the mountain $M_q$.

Construct a solution to the instance $\cC$ by including $i_1, i_2, \ldots, i_q$; 
and for each copy of a wide resource $i$ in $R_2'$,
include a copy of its corresponding long resource. Notice that this is a feasible solution to the instance $\cC$.
The cost of the short resources $\{i_1, i_2, \ldots, i_q\}$ is at most $8$ times the cost of $R_1$
and the cost of the long resources is the same as that of $R_2'$, which is at most twice that of $R_2$.
Cost of $\opt$ is the sum of costs of $R_1$ and $R_2$.
Hence, cost of the constructed solution is at most $8$ times the cost of $\opt$.

We now prove the second property:
Let $S$ be a given a solution to the instance $\cC$ of the {\lspc} problem of cost $c$;
the solution also provides a coverage profile, $k_t$ for each timeslot $t$ (such that $\sum_t k_t = k$).
We produce a feasible solution $S'=(R',J')$ to the instance $\cB$ with the same cost $c$. 
For each long resource picked by $S$, we retain the corresponding
wide resource in $R'$ (maintaining the number of copies).
Consider any timeslot $t$ in the {\lspc} instance and let $M$ be the corresponding mountain in the instance $\cB$.
The solution $S$ contains at most one short resource $i_s$ active at $t$ of capacity $k_t'=w(i_s)$.
Consider the multiset of short resources $R'$ in the instance $\cB$ associated with the resource $i_s$.
The multiset $R'$ covers a set of $k_t'$ jobs contained in the mountain $M$.
Include all these $k_t'$ jobs in $J'$. Choose any other $k_t-k_t'$ jobs contained in $M$
and add these to $J'$; notice that the wide resources retained in $R'$ can cover these jobs.
This way we get a solution $S'$ for the instance $\cB$.
Cost of the solution $S'$ is at most the cost of $S$.
\\

\noindent
{\it Proof of Theorem \ref{thm:xCCC}: }
By composing the reductions given the two steps,
we get a reduction from the {\PResAll} problem on a single mountain range to the {\lspc} problem. 
The first step and the second step incur a loss in approximation of $3$ and $8$, respectively.
Thereby, the combined reduction incurs a loss of $24$.
Theorem \ref{thm:xEEE} provides a $16$-approximation algorithm for the {\lspc} problem.
Combining the reduction and the above algorithm, we get an algorithm for the {\PResAll} for a single mountain
range with an approximation ratio of $16\times 24 = 384$.

\section{Details for {\lspc} Algorithm}
\label{sec:lspcdetails}\
In this section, we present proofs and other details omitted in the main body of the paper.

\subsection{Proof of Lemma \ref{lem:SLRA}}
The following lemma is a reformulation of Theorem 1 {in} \cite{esa2011}.
For a multiset of resources $S$, let $c(S)$ denote its cost.
\begin{lemma}
\label{lem:esa-SLRA}
Let $\wh{S}$ be a multiset of long resources covering a profile $\wh{Q}$ over a timerange $[1,T]$.
Then, there exists a multiset of long resources $S'$ such that $S'$ is a SLRA cover for $Q$
and $c(S')\leq 16\cdot c(\wh{S})$. 
\end{lemma}

Let $\opt$ be the optimum solution and let $Q$ be the profile of measure $k$ covered by $\opt$.
Let $\opt_l$ and $\opt_s$ be the multiset of long and short resources contained in $\opt$, respectively.
Define $Q_l$ to be the residual profile over $[1,T]$: $Q_l(t)=Q(t)- Q^{\sh}_S(t)$.
The multiset $\opt_l$ covers the profile $Q_l$. 
Invoke Lemma \ref{lem:esa-SLRA} on $\opt_l$ and $Q_l$ (taking $\wh{S}=\opt_l$ and $\wh{Q}=Q_l$)
and obtain a multiset of long resources $S'$ which forms a SLRA cover for $Q_l$. 
Construct a new multiset $S$, by taking the union of $S'$ and $\opt_s$. 
Notice that $S$ is a SLRA solution.  The cost  of $S'$ is at most 16 times the cost of $\opt_l$.
So, $S$ has cost at most 16 times the cost of $\opt$.
%This proves the lemma.

\subsection{Proof of Lemma \ref{lem:decomp}}
We first extend the notion of SLRA covers to subsets of timeslots.
Let $\calT\subseteq [1,T]$ be a set of timeslots and let $\wh{Q}$ be a profile over the set $\calT$.
A good multiset of resources $S$ is said to be a SLRA cover for $\calT$, 
if for any timeslot $t\in \calT$, there exists a long resource $i\in S$ such that 
$w(i)f_{S}(i)\geq Q(t)-Q^{\sh}_{S}(t)$.
The following lemma is a reformulation of a result in \cite{esa2011} (see Section 2.2 therein).

\begin{lemma}
\label{lem:esa-timecut}
Let $\wh{S}$ be a multiset consisting of only long resources. 
Let $\wh{Q}$ be a profile over a non-empty set of timeslots $\calT'\subseteq [a,b]$,
for some $a$ and $b$. 
Suppose $\wh{S}$ is a SLRA cover for $\wh{Q}$. Then one of the following properties is true:
\begin{itemize}
\item
There exists a timeslot $t^*\in [a,b-1]$ and a partition of $\wh{S}$ into $\wh{S}_1$ and $\wh{S}_2$ such that
$\wh{S}_1$ is a SLRA cover for $\wh{Q}_1$ and $\wh{S}_2$ is a SLRA cover for $\wh{Q}_2$,
where $\wh{Q}_1$ and $\wh{Q}_2$ are the profiles obtained by restricting $\wh{Q}$
to the timeslots in $\calT'\cap [a,t^*]$ and $\calT'\cap [t^*+1,b]$, respectively.
\item
There exists a resource $i^*\in \wh{S}$ spanning all timeslots in $\calT'$.
\end{itemize}
\end{lemma}

We now prove Lemma \ref{lem:decomp}. 
Consider a good multiset of resources $S$ forming a $h$-free SLRA cover for a profile $Q$ over a range $[a,b]$.
Define the set of timeslots $\calT'$:
\[
\calT' = \{t\in [a,b]~:~Q(t)-Q^{\sh}_S(t)>h\}.
\]
If $\calT'$ is empty, then $S$ is a $h$-free cover for $Q$; 
this corresponds to the first case of Lemma \ref{lem:decomp}.
So, assume $\calT'\neq \emptyset$.
Define a profile $\wh{Q}$ over the timeslots in $\calT'$: for any $t\in \calT'$,
let $\wh{Q}(t) = Q(t)-Q^{\sh}_S(t)$.
Notice that $S$ is a SLRA cover for the profile $\wh{Q}$. 
Invoke Lemma \ref{lem:esa-timecut} (with $\wh{S}=S$)
Let us analyze the two cases of the above lemma. 
Consider the first case in Lemma~\ref{lem:esa-timecut}. 
In this case, there exists a timeslot $t^*$ and a partitioning of $S$ into $S_1$ and $S_2$,
with the stated properties.
In this case, we see that $S_1$ and $S_2$ are $h$-free SLRA covers
for $[a,t^*]$ and $[t^*+1,b]$, respectively. This corresponds to the second case of Lemma \ref{lem:decomp}.
Consider the second case in Lemma~\ref{lem:esa-timecut}. In this case, there exists a long resource $i^*\in S$ such that $i^*$ spans
all the timeslots in $\calT'$. This means that any $t\in [a,s(i^*)-1]$ or $t\in [e(i^*)+1,b]$,
$Q(t)-Q^{\sh}_S(t) \leq h$. This corresponds to the third case of Lemma \ref{lem:decomp}.
\qed

\subsection{Correctness of the Recurrence Relation (Figure \ref{fig:formula})}
\label{sec:recur-proof}
We prove the lemma by induction on the position in which a triple appears in the topological ordering $\pi$.
The base case corresponds to triples that do not have a parent in $G$. The lemma is trivially true in this case.

Consider any triple $z=([a,b],q,h)$. Let $S$ be the optimum $h$-free SLRA $q$-cover for $[a,b]$.
There exists a profile $Q$ over $[a,b]$ such that $Q$ has measure $q$ and $S$ is a $h$-free SLRA cover
for $Q$. Let us invoke Lemma \ref{lem:decomp} and consider its three cases.

Suppose the first case of the lemma is true.
Let $S_s$ be the set of short resources contained in $S$. 
Then, $S_s$ is a $h$-free cover for $Q$. Therefore $E_1=A([a,b],q,h)\leq c(S_s) \leq c(S)$.

Suppose the second case of the lemma is true.
Let $t^*$ be the timeslot and $S_1$ and $S_2$ be the partition given by the lemma.
Let $Q_1$ and $Q_2$ be the profiles obtained by restricting $Q$ to the timeranges $[a,t^*]$ and $[t^*+1,b]$,
respectively. Let the measures of $Q_1$ and $Q_2$ be $q_1$ and $q_2$, respectively.
Then $S_1$ is a $h$-free $q_1$-cover for $[a,t^*]$ and $S_2$ is a $h$-free $q_2$-cover for $[t^*+1,b]$.
Therefore, by induction, $M([a,t^*],q_1,h)\leq c(S_1)$ 
and $M([t^*+1,b],q_2,h)\leq c(S_2)$.
In computing the quantity $E_2$, we try all possible ways of partitioning the range $[a,b]$ and dividing the number $q$.
Hence, $E_2\leq c(S_1)+c(S_2)$. Since $c(S)=c(S_1)+c(S_2)$, we see that $E_2\leq c(S)$.

Suppose the third case of lemma is true.
Let $i^*$ be the long resource given by the lemma.
Let $S_1$ be the set of short resources contained in $S$ that are active at a timeslot $t\in[a,s(i^*)-1]$.
Similarly, let $S_3$ be the set of short resources contained in $S$ that are active at a timeslot $t\in[e(i^*)+1,b]$.
Let $S_2$ be the multiset of long resources contained in $S$ and the set of short resources contained in $S$
that are active at a timeslot $t\in [a,b]$.
Let $Q_1$, $Q_2$ and $Q_3$ be the profiles obtained by restricting $Q$ to the ranges $[a,s(i^*)-1]$,
$[s(i^*),e(i^*)]$ and $[e(i^*)+1,b]$, respectively.
The lemma guarantees that $S_1$ and $S_2$ are $h$-free covers for $Q_1$ and $Q_3$ respectively.
Let $q_1$, $q_2$ and $q_3$ be the measures of $Q_1$, $Q_2$ and $Q_3$, respectively.
We see that $A([a,s(i^*)+1],q_1,h)\leq c(S_1)$ and $A([e(i^*)+1,b],q_3,h)\leq c(S_3)$.
Let $\alpha^*=f_S(i^*)$ be the number of copies of $i^*$ present in $S$.
Notice that if $\alpha^* w(i^*)\leq h$, then $i^*$ is not a useful resource,
because $i^*$ will be covering only timeslots in $[s(i^*),e(i^*)]$ with residual demands at most $h$;
but all such timeslots are free and need not be covered.
So, without loss of generality, assume that $\alpha^* w(i^*)>h$.
Since $i^*$ spans the entire range $[s(i^*),e(i^*)]$, 
the resource $i^*$ can cover all timeslots in the above range with residual demands at most $\alpha^* w(i^*)$.
Let $S_2'=S_2-\{i^*\}$. Notice that $S_2'$ is a $(\alpha^* w(i))$-free SLRA cover for the profile $Q_2$.
Therefore, $S_2'$ is a $(\alpha^* w(i))$-free $q_2$-cover for the range $[s(i^*),e(i^*)]$.
Hence, by induction, $M([s(i^*),e(i^*)],q_2,\alpha^* w(i^*)) \leq c(S_2')$.
Therefore, $E_3 \leq c(S_1)+c(S_2)+c(S_3)=c(S)$.

The quantity $E= \min \{E_1, E_2, E_3\}$; so 
$E\leq c(S)$. The proof is now complete.
\qed
\end{document}